\documentclass{article}

\usepackage{amsfonts}%
\usepackage{amsmath}%
\usepackage{amssymb}%
\usepackage{graphicx}
\usepackage{appendix}
\usepackage{color}

\newtheorem{theorem}{Theorem}

\newtheorem{corollary}{Corollary}

\newtheorem{definition}{Definition}
\newtheorem{example}{Example}

\newtheorem{lemma}{Lemma}

\newtheorem{remark}{Remark}

\newenvironment{proof}[1][Proof]{\noindent\textbf{#1.} }{\ \rule{0.5em}{0.5em}}

\begin{document}

\title{Eigenvectors in terms of reduced complements \\of minor determinants}
\author{M.I. Krivoruchenko\footnote{E-mail: mikhail.krivoruchenko@itep.ru} \\\\
\textit{National Research Centre "Kurchatov Institute"} \\ \textit{Pl. Akademika Kurchatova 1, 123182 Moscow, Russia}}



\maketitle

\begin{abstract}
Eigenvectors associated with non-degenerate eigenvalues are shown to correspond to columns of the adjugate of the characteristic matrix. Degenerate eigenvalues are associated with eigenvectors that correspond to reduced complement tensors of minor determinants of the characteristic matrix. These observations are corroborated by a description of the non-degenerate two-level system and the Dirac equation, which exhibits twofold spin degeneracy of energy eigenvalues. Trace identities for the reduced order-one complement tensor and the diagonal sum of minor determinants are also presented.
\\
\\
\textbf{Keywords:} eigenvalue problem, eigenvectors, complements of minors \\
\textbf{PACS number:} 15A18
\end{abstract}

\section{Introduction}
\renewcommand{\theequation}{1.\arabic{equation}}
\setcounter{equation}{0}

Eigenvalue equation 
\begin{equation} \label{eq}
\mathbf{H}\psi=\lambda\psi 
\end{equation}
is applicable to all domains of the natural sciences that use linear algebra, including classical mechanics, quantum physics, quantum chemistry, mathematical statistics, etc.
In this equation, $\mathbf{H}$ is an $n \times n$ matrix with complex elements, and $\psi$ 
is a vector of the complex space $\mathbb{C}^n$. 
The eigenvalues $\lambda$ are identified as roots of the characteristic equation 
\begin{equation}
\det\left( \lambda \mathbf{I} -\mathbf{H}\right) =0,  \label{det = 0}
\end{equation}
where $\mathbf{I}$ is the $n \times n$ identity matrix.

A necessary and sufficient condition for the existence of $n$ nonzero eigenvectors 
$\psi \in \mathbb{C}^n$ is the fulfillment of equality (\ref{det = 0}). 
The characteristic matrix 
\begin{equation}
\mathbf{C}(\lambda) = \lambda \mathbf{I} -\mathbf{H}
\end{equation}
shows linear dependence in its rows, 
when the eigenvalues are determined and $\mathbf{C}(\lambda)$ has a rank of less than $n$:    
$\mathrm{rank}( \mathbf{C}(\lambda) ) <n$. 

The characteristic equation is a polynomial equation of degree $n$. 
Faddeev and Sominskii \cite{Faddeyev:1949} and J. S. Frame \cite{Frame:1949} 
construct a recursive algorithm for calculating characteristic polynomial coefficients in terms of the traces of the powers of the matrix.
The techniques are inspired by Le Verrier's paper \cite{LeVerrier:1840}. 
A closed analytical form for the recursion solutions \cite{Reed:1978,Kondratyuk:1992,Brown:1994,Krivoruchenko:2016} can be obtained by highlighting specific aspects of matrix algebra based on Cayley-Hamilton's theorem \cite{Gantmacher:1976}.
Trace identities are also known to exist for adjugate matrices \cite{Reed:1978,Kondratyuk:1992,Krivoruchenko:2016} and Pfaffians \cite{Krivoruchenko:2016}. Other noteworthy 
developments include the analytical relationship between modules of eigenvector components and
diagonal minor determinants and eigenvalues \cite{Thompson:1966} (for a recent review, see \cite{Denton:2021}), 
and particular results for mass matrices 
in the theory of Majorana neutrinos
\cite{Krivoruchenko:2024}.

Advanced numerical approaches for calculating eigenvalues and 
eigenvectors include the Arnoldi iteration \cite{Arnoldi:1951}, the QR algorithm \cite{Francis:1961,Francis:1962,Kublanovskaya:1962} and others (for a review, see \cite{Trefethen:1997}). 

In this paper, we show that  in a scenario of $\mathrm{rank}(\mathbf{C}(\lambda)) = n - 1$, the adjugate of $\mathbf{C}(\lambda)$ determines the eigenvectors of the eigenvalue equation (Sect. 3, Theorem 1). For an $s$-fold degeneracy of the eigenvalues and $\mathrm{rank}(\mathbf{C}(\lambda)) < n - 1$,
the reduced complement tensor to an order-$s$ minor determinant of $\mathbf{C}(\lambda)$ specifies $s$ linear independent eigenvectors (Sect. 4, Theorem 3). We establish components of the reduced complement tensor that match the eigenvectors. Trace identities for the reduced complement tensors (Sect. 4, Theorem 4) and the diagonal sum of minor determinants (Appendix B, Theorem 6) are also derived. 
Each of the aforementioned new theorems is backed by proofs of the most relevant 
previously known facts.
Two examples of using the proposed methods to describe the non-degenerate two-level system (Sect. 3, Example 1) and the Dirac equation with twofold spin degeneracy of energy eigenvalues (Sect. 4, Example 3) are considered.




\section{Basic notations and definitions}
\renewcommand{\theequation}{2.\arabic{equation}}
\setcounter{equation}{0}
 
Contravariant and covariant components of tensors in the
coordinate system, which diagonalizes matrix $\mathbf{H}$, are represented
by Gothic letters $\mathfrak{a}$, $\mathfrak{b}$, \ldots\ from the beginning
of the alphabet. Eigenvalues and eigenvectors are
numbered by covariant indices: $\lambda_{\mathfrak{a}}$ and $\psi_{\mathfrak{a}}$ for 
$\mathfrak{a} \in \Sigma_{n} = (1,\ldots,n).$ Components of tensors in the initial
coordinate system are represented by Latin letters from the middle of the alphabet: $i$, $%
j$, \ldots. The eigenvectors of $\mathbf{H}$ are the eigenvectors of $\mathbf{C}(\lambda)$, and 
vice versa.

We discuss the main scenario for a Hermitian $n \times n$ matrix $\mathbf{H}$. The characteristic matrix $\mathbf{C}(\lambda)$ is also Hermitian. The unitary matrices transforming the coordinate systems are denoted by $\mathbf{U}$ and $\mathbf{V}$.
The matrices in statements that are independent of the type of matrix are denoted by $\mathbf{A}$, $\mathbf{B}$, and $\mathbf{D}$. 

\begin{definition}
The adjugate of an $n \times n$ marix $\mathbf{A}$ is defined by
\begin{equation}
\left( \mathrm{adj}\, \mathbf{A}\right)_{i}^{j}=
\frac{1}{(n-1)!}\delta_{ii_{2}\ldots i_{n}}^{jj_{2}\ldots j_{n}}
A_{j_{2}}^{i_{2}}\ldots A_{j_{n}}^{i_{n}}.  \label{adj def}
\end{equation}
\end{definition}
The summation is performed on the identical indices in accordance with Einstein's summation rule. Independent summation over identical indices is the same as summing over ordered sets of identical indices. the basic properties of the generalized Kronecker symbol entering the right side of the equation are described in Appendix A.  

In accordance with Eq. (\ref{eq}), the eigenvectors $\psi$ and rows of the
characteristic matrix are orthogonal,
when the complex space is endowed with a Euclidean metric.
Given that the characteristic matrix $\mathbf{C}(\lambda)$ is
degenerate, one can find a vector $\psi\neq 0$ orthogonal to all its rows.

To construct $\psi$, we first consider the identity 
\begin{equation}
A_{k}^{j} \left( \mathrm{adj}\, \mathbf{A}\right)  
_{i}^{k}=\delta_{i}^{j} \det  \mathbf{A}
\label{basic eq adj}
\end{equation}
for the adjugate matrix $\mathbf{A}$.
When the determinant in Eq.~(\ref{basic eq adj}) vanishes, we obtain a family
of formal solutions to Eq. (\ref{eq}):
\begin{equation}
\psi_{i}^{j}(s,\lambda) = \left( \mathrm{adj} \, \mathbf{A} \right)
_{i}^{j}  \label{eigenv}
\end{equation}
with $\mathbf{A} = \mathbf{C}(\lambda)$ and $i \in \Sigma_{n}$. 
The meaning of argument $s$ is explained below.
An essential parameter is also the algebraic multiplicity $\mu_{\mathbf{H}}(\lambda)$ of the eigenvalue $\lambda$. 
In the next section, we show that
for a rank $n - 1$ characteristic matrix, all vectors (\ref{eigenv}) with lables $i \in \Sigma_{n}$ 
are parallel to one another and define a unique eigenvector.

We adhere to 
\begin{definition}
    The minor determinant of order $s \leq n$ of an $n \times n$ matrix $\mathbf{A}$ is an antisymmetric tensor
    \begin{eqnarray}
    \left(\det \mathbf{A} \right)^{j_{1} \ldots j_{s}}_{i_{1} \ldots i_{s}} &=& A_{[i_{1}}^{j_{1}}\ldots A_{i_{s}]}^{j_{s}} \nonumber \\
    &=& \frac{1}{s!}
        \delta_{k_{1} \ldots k_{s}}^{j_{1} \ldots j_{s}} 
        \delta_{i_{1} \ldots i_{s}}^{l_{1} \ldots l_{s}} 
        A_{l_{1}}^{k_{l}}\ldots A_{l_{s}}^{k_{s}} \label{def 3}
    \end{eqnarray}
    for $(j_{1},\ldots,j_{s}) \subseteq \Sigma_{n}$ and $(i_{1},\ldots,i_{s}) \subseteq \Sigma_{n}$. The antisymmetrization in the covariant indices is assumed in the first line and made explicit in the second line. 
\end{definition}

\begin{definition}
The order of a tensor is the number of its contravariant and covariant indices. The order of an antisymmetric tensor with the same number of contravariant and covariant indices is the number of its contravariant indices.
\end{definition}

For $s=n$, the minor determinant equals 
\begin{eqnarray}
    \left(\det \mathbf{A} \right)^{j_{1} \ldots j_{n}}_{i_{1} \ldots i_{n}} &=& A_{[i_{1}}^{j_{1}}\ldots A_{i_{n}]}^{j_{n}} \nonumber \\ &=& \delta_{i_{1} \ldots i_{n}}^{j_{1} \ldots j_{n}} \det \mathbf{A}.
    \label{det}
\end{eqnarray}

\begin{definition} 
Let $\mathbf{A}$ be an $n\times n$ matrix. The antisymmetric tensor 
\begin{equation} \label{complement def}
(\underset{s}{\det } \,\mathbf{A})_{i_{1}\ldots i_{r}}^{j_{1}\ldots
j_{r}} 
= \frac{1}{(n - s)!}
\delta\underset{r + n - s}{\underbrace{_{_{{i_{1}\ldots i_{r}i_{s+1}\ldots i_{n}}}}}
    ^{j_{1} \ldots j_{r}j_{s+1} \ldots j_{n}}}
\underset{n-s}{\underbrace{A_{j_{s+1}}^{i_{s+1}}\ldots A_{j_{n}}^{i_{n}}}}
\end{equation}
with $r \leq s$ and index sets $J=(j_{1},\ldots ,j_{r})\subseteq \Sigma _{n}$ and $%
I=(i_{1},\ldots ,i_{r})\subseteq \Sigma _{n}$ is called a reduced complement
of order $r,s$ to the minor determinants of order $s\leq n$. The
index sets $(j_{s+1},\ldots ,j_{n})=\Sigma _{n}\setminus J$ and $%
(i_{s+1},\ldots ,i_{n})=\Sigma _{n}\setminus I$ determine columns and rows of the matrix $\mathbf{A}$ used to construct the minor determinants. The reduced complement of order $s,s$ 
is called a complement of order $s$.
\end{definition}

In particular, $(\underset{0}\det \, \mathbf{A})_{{I}}^{{J}}=\det \mathbf{A}$, where 
${J}={I}=\varnothing$, and  
$(\underset{1}\det \mathbf{A})_{i}^{j}=(\mathrm{adj}\,\mathbf{A})_{i}^{j}$. 
The antisymmetrized products of matrices $\mathbf{A}$ in Eq. (\ref{adj def}) are the minor determinants of order $n - 1$, while 
the right side is the order-one complement.
The first argument of $\psi_{i}^{j}(s,\lambda)$ defined in Eq. (\ref{eigenv}) is the size of the matrix 
$\mathbf{C}(\lambda)$
minus the number of matrices entering the product in Eq.~(\ref{adj def}). 
Given subsequent generalizations in Sect. 4, the parameter $s$ is not always equal to one.

\begin{remark} The reduced complement of order $r,s$ is the complement of order $s$
subjected to $s-r$ contractions of the identical indices:
\begin{equation} \label{complement def trace}
(\underset{s}{\det } \, \mathbf{A})_{i_{1}\ldots i_{r}}^{j_{1}\ldots
j_{r}} 
=
\frac{1}{(s-r)!}(\underset{s}{\det } \,\mathbf{A})
\underset{s}{\underbrace{_{i_{1}\ldots i_{r}j_{r+1}\ldots j_{s}}
^{j_{1} \ldots j_{r}j_{r+1}\ldots j_{s}}}}.
\end{equation}
\end{remark}
This equation can be verified using Eqs. (\ref{complement def}) and (\ref{APP 3}). 

Each minor determinant of order $n-s$ makes a contribution to Eqs. (\ref{complement def}) and (\ref{complement def trace}) with a weight of $\pm 1$ or zero. 
In many instances, this observation enables to recover the combinatorial coefficients.

We use the terms "minor" and "minor determinant" and "complement" and "complement tensor" as synonyms, respectively.

The $\mathbf{H}$ matrix 
can be diagonalized using a unitary matrix $\mathbf{U}$: 
\begin{equation}
\mathbf{H}=\mathbf{U} \mathbf{H}^{\prime} \mathbf{U}^{-1}.
\nonumber 
\end{equation}
Here, $\mathbf{H}^{\prime} =\mathrm{diag}(\lambda_{1},\ldots,\lambda_{n})$ represents a matrix with eigenvalues of $\mathbf{H}$ on the main diagonal.
Upon the transformation, the characteristic matrix becomes
\begin{equation}
\mathbf{C}^{\prime}(\lambda) =\mathrm{diag}(\lambda-\lambda_{1},\ldots,\lambda-\lambda_{n}).    
\nonumber %
\end{equation}
Vectors of Hilbert space $\mathbb{C}^n$ are transformed accordingly: %
$\psi = \mathbf{U} \psi^{\prime}$,  
so that Eq.~(\ref{eq}) can be written as follows:
\begin{equation}
\mathbf{C}^{\prime} (\lambda) \psi^{\prime} = 0.  \label{tr3}
\end{equation}

In matrix notation, vectors and tensors with two indices are marked by a stroke in a coordinate system diagonalizing $\mathbf{H}$. 
When tracking differences at the index level, tensor notation omits the stroke. 
The transformation matrix is denoted accordingly by $U^{j}_{\mathfrak{a}}$, and its inverse is  
$(\mathbf{U}^{-1})^{\mathfrak{a}}_{i} = 
(\mathbf{U}^{\dagger})^{\mathfrak{a}}_{i} = 
(\mathbf{U}^{\ast})^{i}_{\mathfrak{a}}
\equiv U^{\mathfrak{a}}_{i}$, 
and $U^{j}_{\mathfrak{a}}  U^{\mathfrak{a}}_{i} = \delta^{j}_{i}$ 
and $U^{\mathfrak{a}}_{j} U^{j}_{\mathfrak{b}}   = \delta^{\mathfrak{a}}_{\mathfrak{b}}$ .

\section{Eigenvectors for non-degenerate eigenvalues}
\renewcommand{\theequation}{3.\arabic{equation}}
\setcounter{equation}{0}

Eigenvalues of matrix $\mathbf{H}$ are assumed to be enumerated
and we consider a fixed index $\mathfrak{d}$. 
The rank of the haracteristic matrix: $\mathrm{rank}(\mathbf{C}(\lambda_{\mathfrak{d}})) = n - 1$, corresponds to the algebraic multiplicity $\mu_{\mathbf{H}}(\lambda_{\mathfrak{d}}) = s = 1$ of non-degenerate eigenvalues, which is the case discussed in this section.
Applying to Eq. (\ref{adj def}) a transformation to the diagonal coordinate
system, where the characteristic matrix equals $C_{\mathfrak{b}}^{\mathfrak{a}}(\lambda) =\delta_{\mathfrak{b}}^{\mathfrak{a}}\left( \lambda-\lambda_{\mathfrak{b}}\right) $, and using Eq.~(\ref{eigenv}), we obtain 
\begin{align}
\psi_{\mathfrak{b}}^{\mathfrak{a}}(s,\lambda_{\mathfrak{d}}) & = {U}_{j}^{\mathfrak{a}}\psi_{i}^{j}(s,\lambda_{\mathfrak{d}})U_{\mathfrak{b}}^{i} \nonumber \\
& =\frac{1}{(n-1)!}\delta^{\mathfrak{aa}_{2} \ldots\mathfrak{a}_{n}}
                         _{\mathfrak{bb}_{2} \ldots\mathfrak{b}_{n}} 
C_{\mathfrak{a}_{2}}^{\mathfrak{b}_{2}}(\lambda_{\mathfrak{d}}) \ldots C_{\mathfrak{a}_{n}}^{\mathfrak{b}_{n}}(\lambda_{\mathfrak{d}})
\nonumber \\
& =\delta_{\mathfrak{b}}^{\mathfrak{a}}\prod\limits_{\mathfrak{c}\neq%
\mathfrak{b}}\left( \lambda_{\mathfrak{d}} -\lambda_{\mathfrak{c}}\right) . \nonumber 
\end{align}
To prevent the right side of the
equation from turning to zero, $\lambda_{\mathfrak{d}}$ should not
appear in the product sign in place of $\lambda_{\mathfrak{c}}$. The only non-vanishing column of the adjugate
matrix ${\psi}_{\mathfrak{b}}^{\mathfrak{a}}$ is obviously the one with $\mathfrak{b=d}$. At the same time, ${\psi}_{%
\mathfrak{b}}^{\mathfrak{a}}$ has a
non-vanishing component at $\mathfrak{a=b=d}$, in accord with 
the statement that for degenerate matrix $\mathbf{A}$
$\mathrm{rank}( \mathrm{adj}\,\mathbf{A} ) \leq 1$ \cite{Prasolov:1994}. The same conclusion 
follows form Eq.~(\ref{tr3}).

Using the reciprocal transformation, the adjugate matrix in the
initial coordinate system can be found (cf. (\ref{eigenv})): 
\begin{equation}
\psi_{i}^{j}(s,\lambda_{\mathfrak{d}}) =U_{\mathfrak{d}}^{j} 
U_{i}^{\mathfrak{d}}\prod\limits_{\mathfrak{c}\neq\mathfrak{d}%
}\left( \lambda_{\mathfrak{d}}-\lambda_{\mathfrak{c}}\right) .
\label{eigen vec equiv}
\end{equation}

The construction of matrix $\mathbf{U}$ is not mandatory; one can
immediately apply the formula (\ref{adj def}) in the initial coordinate
system. The equivalent representation (\ref{eigen vec equiv}) states that
for any eigenvalue, $\lambda=\lambda_{\mathfrak{d}}$, the factor $ 
U_{i}^{\mathfrak{d}}$ alters the norm only, hence all
eigenvectors $\psi_{i}^{j}$ with $i \in \Sigma_{n}$
are parallel and give the same solution.

The eigenvectors obtained in the formalism constitute an orthogonal set in
terms of the Hermitian product. Let $\lambda_{\mathfrak{b}}$ and $\lambda_{\mathfrak{d}}$ 
be eigenvalues corresponding to the eigenvectors $\psi_{i }^{j}(\lambda_{\mathfrak{b}})$ and $\psi_{k }^{j}(\lambda_{\mathfrak{d}})$, respectively.  The Hermitian product equals
\begin{align}
\sum_{j}\psi_{i}^{j\ast}(s,\lambda_{\mathfrak{b}}) \psi_{k}^{j}(s,\lambda_{\mathfrak{d}}) 
& = \sum_{j} {U}_{j}^{\mathfrak{b}}U_{\mathfrak{d}}^{j} \, {U}_{\mathfrak{b}}^{i} {U}_{k}^{%
\mathfrak{d}}\prod\limits_{\mathfrak{c}\neq\mathfrak{b}}\left( \lambda_{%
\mathfrak{b}}-\lambda_{\mathfrak{c}}\right) \prod \limits_{\mathfrak{e}\neq%
\mathfrak{d}}\left( \lambda_{\mathfrak{d}}-\lambda_{\mathfrak{e}}\right) \nonumber \\
& =\delta_{\mathfrak{d}}^{\mathfrak{b}}\,{U}_{\mathfrak{b}}^{i} {U}_{k}^{\mathfrak{d}} \prod\limits_{\mathfrak{c}\neq\mathfrak{b}}\left( \lambda_{\mathfrak{b}%
}-\lambda_{\mathfrak{c}}\right) ^{2}. \nonumber 
\end{align}
The norm for $\mathfrak{b} = \mathfrak{d}$ is positively defined, so the set of
eigenvectors can be made orthonormal.

We have thus proven
\begin{theorem}
    Let $\mathbf{H}$ be an $n \times n$ Hermitian matrix and let $\lambda_{\mathfrak{d}}$ be an eigenvalue of $\mathbf{H}$ of algebraic multiplicity $\mu_{\mathbf{H}}(\lambda_{\mathfrak{d}}) = 1$. 
    In the diagonal coordinate system, the adjugate matrix $\psi_{\mathfrak{b}}^{\mathfrak{a}}(s = 1,\lambda_{\mathfrak{d}})$ of the characteristic matrix $\mathbf{C}^{\prime}(\lambda_{\mathfrak{d}})$
    has a unique non-vanishing component 
    at $\mathfrak{a} = \mathfrak{b} = \mathfrak{d}$. 
    The adjugate matrix for fixed values of $\mathfrak{b}=\mathfrak{d}$ determines an eigenvector of  $\mathbf{H}^{\prime}$ whose contravariant components are labeled with the index $\mathfrak{a}$.
    Up to normalization, the $n$ vectors connected to the adjugate matrix 
    $\psi_{i}^{j}(s = 1,\lambda_{\mathfrak{d}})$ in the initial coordinate system  for $i\in\Sigma_{n} = (1,\ldots,n)$ coincide since they are parallel to one another.
    Each of them can be selected as the appropriate eigenvector. The eigenvectors for the entire set of eigenvalues of matrix $\mathbf{H}$  are orthogonal and can be normalized to unity.
\end{theorem}


\begin{example}
Two-level system
\end{example}

A two-level system is an excellent illustration of specific characteristics 
of quantum mechanical systems, which have applications in solid-state physics, elementary particle physics to describe oscillations of kaons and neutrinos, quantum chemistry to model avoiding crossing phenomena, and many others. 
A two-level system is discussed 
in Appendix C of the monograph \cite{Bilenky:2018}, 
Sect. 2.1 of Ref. \cite{Krivoruchenko:2011}, 
and Sect. III of Ref. \cite{Blaum:2020}; Task 1 to Sect. 39 of the monograph \cite{Landau:1987} also serves as an excellent guideline. 

\subsection*{\small {a). An ad hoc solution}}

The Hamiltonian $\mathbf{H}$ can be represented by a $2\times 2$ matrix as the
sum of a matrix proportional to the identity matrix $\mathbf{I}$ and terms
proportional to the Pauli matrices: 
\begin{equation}
\mathbf{H}=a\mathbf{I}+b\mbox{\boldmath$\sigma$}\mathbf{n}.  \label{A 1}
\end{equation}
Here, $a$ and $b$ are real numbers, $\mathbf{n}$ is a real unit vector, and $%
\mathbf{H}$ is a Hermitian matrix. Given that $(\mbox{\boldmath$\sigma$}\mathbf{n})^2 = 1$, $\mathbf{H}$'s eigenvalues equal  
\begin{equation}
\lambda _{\pm }=a\pm b.  
\end{equation}
Projection operators 
\begin{equation}
P_{\pm }=\frac{1\pm \mbox{\boldmath$\sigma$}\mathbf{n}}{2}
\end{equation}
act on a reference spinor, which we select to be $(1,0)^{\mathrm{T}}$, to yield  eigenvectors
\begin{equation}
\left( 
\begin{array}{l}
\chi _{\pm }^{1} \\ 
\chi _{\pm }^{2}%
\end{array}%
\right) =N_{\mathfrak{\pm }}P_{\pm }\left( 
\begin{array}{l}
1 \\ 
0%
\end{array}%
\right) .  \label{eigenvectors}
\end{equation}%
The normalization factors 
$N_{\pm }$ are determined from the equation $\chi _{\pm }^{\dagger }\chi
_{\pm }=1$, while orthogonality $\chi _{\pm }^{\dagger }\chi _{\mp }=0$ is provided by $P_{+}P_{-}=0$. 

In terms of traces of $\mathbf{H}$, $\mbox{\boldmath$\sigma$}\mathbf{H}$ and $\mathbf{H}^{2}$ the parameters $a$, $b$, and $\mathbf{n}$  are determined by the equations
\begin{equation}
\frac{1}{2}\mathrm{Tr}[\mathbf{H}]=a,\ \ \ \ \frac{1}{2}\mathrm{Tr}[\mathbf{H%
}^{2}]=a^{2}+b^{2}  \nonumber 
\end{equation}%
and%
\begin{equation}
\mathbf{n}=
\frac{\frac{1}{2}\mathrm{Tr}[\mbox{\boldmath$\sigma$}\mathbf{H}]}
     {\sqrt{\frac{1}{2}\mathrm{Tr}[\mathbf{H}^{2}]-\frac{1}{4}\mathrm{%
Tr}[\mathbf{H}]^{2}}}.  \nonumber 
\end{equation}

Let us examine a two-level system, assuming the most commonly used parametrization
\begin{equation}
\mathbf{H}=\left\Vert 
\begin{array}{ll}
V_{11} & V_{12} \\ 
V_{21} & V_{22}%
\end{array}%
\right\Vert , 
\end{equation}%
where $V_{11}$ and $V_{22}$ are real, $V_{12}$ and $V_{21}=V_{12}^{\ast }$
are complex. 

The eigenvalues of $\mathbf{H}$ are as follows: 
\begin{equation}
\lambda _{\pm }=\frac{V_{11}+V_{22}\pm \omega }{2},
\end{equation}
where 
\begin{equation*}
\omega =\sqrt{\left( V_{11}-V_{22}\right) ^{2}+4V_{12}V_{21}}.
\end{equation*}
Projection operators take the form 
\begin{equation}
P_{\pm }=\left\Vert 
\begin{array}{cc}
\frac{1}{2}\pm \frac{V_{11}-V_{22}}{2\omega } & \pm \frac{V_{12}}{\omega }
\\ 
\pm \frac{V_{21}}{\omega } & \frac{1}{2}\mp \frac{V_{11}-V_{22}}{2\omega }%
\end{array}%
\right\Vert . 
\end{equation}%

Using Eq.~(\ref{eigenvectors}), one finds the eigenvectors: 
\begin{equation}
\chi _{\pm }=\left( 
\begin{array}{l}
\left( \frac{V_{21}^{\ast }}{|V_{21}|}\right) ^{1/2}\left( \frac{1}{2}\pm 
\frac{V_{11}-V_{22}}{2\omega }\right) ^{1/2} \\ 
\pm \left( \frac{V_{21}}{|V_{21}|}\right) ^{1/2}\left( \frac{1}{2}\mp \frac{%
V_{11}-V_{22}}{2\omega }\right) ^{1/2}%
\end{array}%
\right) .  \label{two level ev}
\end{equation}%

\subsection*{{\protect\small {b). Eigenvectors in terms of adjugate matrix}}}

The characteristic matrix is equal to 
\begin{equation*}
\mathbf{C}(\lambda _{\pm })=\lambda _{\pm }-\mathbf{H=}\left\Vert 
\begin{array}{cc}
\frac{-V_{11}+V_{22}\pm \omega }{2} & -V_{12} \\ 
-V_{21} & \frac{V_{11}-V_{22}\pm \omega }{2}%
\end{array}%
\right\Vert .
\end{equation*}%
%

Trace identities for a $2\times 2$ matrix $\mathbf{A}$
give \cite{Kondratyuk:1992,Krivoruchenko:2016} 
\begin{equation} \label{TI}
\mathrm{adj}\ \mathbf{A}=\mathrm{Tr}\left[ \mathbf{A}\right] \mathbf{I}-%
\mathbf{A}.
\end{equation}%
The eigenvectors $\psi _{i}^{j}(1,\lambda _{\pm })$ paired with the
corresponding eigenvalues $\lambda =\lambda _{\pm }$ can be found 
with the use of Eqs.~(\ref{eigenv}) and (\ref{TI}) to give
\begin{eqnarray}
\left\Vert 
\begin{array}{cc}
\psi _{1}^{1}(1,\lambda _{\pm }) & \psi _{2}^{1}(1,\lambda _{\pm }) \\ 
\psi _{1}^{2}(1,\lambda _{\pm }) & \psi _{2}^{2}(1,\lambda _{\pm })%
\end{array}%
\right\Vert &=&\mathrm{adj}\ \mathbf{C}(\lambda _{\pm })  \nonumber \\
&=&\left\Vert 
\begin{array}{cc}
\frac{V_{11}-V_{22}\pm \omega }{2} & V_{12} \\ 
V_{21} & \frac{-V_{11}+V_{22}\pm \omega }{2}%
\end{array}%
\right\Vert .
\end{eqnarray}

The area of a parallelogram stretched along the eigenvectors $%
\psi_{i}^{j}(1,\lambda_{\pm})$ is zero: 
\begin{equation*}
\psi_{1}^{j}(1,\lambda_{\pm}) \psi_{2}^{k}(1,\lambda_{\pm}) -
\psi_{1}^{k}(1,\lambda_{\pm}) \psi_{2}^{j}(1,\lambda_{\pm}) = 0, 
\end{equation*}
so they are parallel. The eigenvectors $\psi_{i}^{j}(1,\lambda_{%
\pm})$ are also parallel to the eigenvectors $\chi^j_{\pm}$, because 
\begin{equation*}
\psi_{i}^{j}(1,\lambda_{\pm}) \chi^k_{\pm} - \psi_{i}^{k}(1,\lambda_{\pm})
\chi^j_{\pm} = 0. 
\end{equation*}
The eigenvectors $\psi_{1}^{j}(1,\lambda_{\pm})$, $\psi_{2}^{j}(1,\lambda_{\pm})$, and $\chi^j_{\pm} $ , when normalized,  coincide up to a phase factor.

In summary, we have demonstrated that the novel method effectively reconstructs solutions to the two-level system.


\section{Eigenvectors for degenerate eigenvalues}
\renewcommand{\theequation}{4.\arabic{equation}}
\setcounter{equation}{0}

To remove the degeneracy of an eigenvalue $\lambda$, one can always introduce a perturbation to the matrix $\mathbf{H}$. 
Given $\mathrm{rank}(\mathbf{C}(\lambda)) = n-1$, the eigenvectors can be constructed using the formalism developed in the previous section.
The crucial issue, however, is whether a formalism covering the same array of concepts 
can be developed and directly applied to systems with degenerate eigenvalues.

The reduced complements of order $r,s$ of minor determinants of the characteristic matrix
\begin{align}
\psi_{i_{1}\mathfrak{\ldots} i_{r}}^{j_{1}\mathfrak{\ldots} j_{r}} (s,\lambda) 
= (\underset{s}\det \, \mathbf{C}(\lambda))^{{j_{1} \ldots j_{r}}}_{i_{1} \ldots i_{r}} 
\label{define wf adj tensor}
\end{align}
are algebraic objects that, in the general scenario, might be associated with the eigenvectors.

\subsection{Basic properties of $\psi_{i_{1}\mathfrak{\ldots} i_{r}}^{j_{1}\mathfrak{\ldots} j_{r}} (s,\lambda)$}

Let pairwise distinct indices $\mathfrak{b}_{1}$, \ldots, $\mathfrak{b}_{s} \in \Sigma_{n}$ enumerate eigenvectors of matrix $\mathbf{H}$ associated to an eigenvalue $\lambda_{\mathfrak{d}}$ of
algebraic multiplicity $\mu_{\mathbf{H}}(\lambda_{\mathfrak{d}}) = s > 1$. 
In the diagonal coordinate system, the expression (\ref{define wf adj tensor}) for $r=s$ simplifies to
\begin{align}
\psi_{\mathfrak{b}_{1}\ldots \mathfrak{b}_{s}}^
     {\mathfrak{a}_{1}\ldots a_{s}}(s,\lambda_{\mathfrak{d}}) & = U_{j_{1}}^{\mathfrak{a}_{1}}\ldots 
U _{j_{s}}^{\mathfrak{a}_{s}}\psi_{i_{1}\mathfrak{\ldots} i_{s}}^{j_{1}\mathfrak{\ldots}%
 j_{s}}(s,\lambda_{\mathfrak{d}})U_{\mathfrak{b}_{1}}^{i_{1}}\ldots U_{\mathfrak{b}%
_{s}}^{i_{s}}  \notag \\
& =\frac{1}{(n-s)!}\delta^{\mathfrak{a}_{1}\ldots  \mathfrak{a}_{s}  \mathfrak{a}_{s + 1} \ldots\mathfrak{a}_{n}}
                         _{\mathfrak{b}_{1}\ldots  \mathfrak{b}_{s}\mathfrak{b}_{s + 1} \ldots\mathfrak{b}_{n}}
C_{\mathfrak{a}_{s+1}}^{\mathfrak{b}_{s+1}}(\lambda_{\mathfrak{d}})\ldots 
C_{\mathfrak{a}_{n}  }^{\mathfrak{b}_{n}}(\lambda_{\mathfrak{d}})  \notag \\
& =\delta_{\mathfrak{b}_{1} \ldots \mathfrak{b}_{s}}
         ^{\mathfrak{a}_{1} \ldots \mathfrak{a}_{s}}
         \prod\limits_{\mathfrak{c}\notin\left( \mathfrak{b}%
_{1},\ldots,\mathfrak{b}_{s}\right) }\left( \lambda_{\mathfrak{d}}-\lambda_{\mathfrak{c}%
}\right) .  \label{r ge 2}
\end{align}
Whenever the set of indices $\mathfrak{a}_{1}$, $\ldots$, $\mathfrak{a}_{s}$ coincides up to a
permutation with the set of indices $\mathfrak{b}_{1}$, $\ldots$, $\mathfrak{b}_{s}$ and these sets both correspond an $s$-fold degenerate eigenvalue $\lambda_{\mathfrak{d}}$ with 
$\mathfrak{d} \in (\mathfrak{b}_{1},\ldots,\mathfrak{b}_{s})$, 
the components of complement tensor
$\psi_{\mathfrak{b}_{1} \ldots \mathfrak{b}_{s}}^{\mathfrak{a}_{1} \ldots \mathfrak{a}_{s}}(\lambda _{\mathfrak{d}})$ do not vanish. 
Equation (\ref{r ge 2}) contains under the
product sign $n-s$ multipliers $\lambda_{\mathfrak{d}} - \lambda_{\mathfrak{c}} \neq 0$. All other components vanish. 

In the initial coordinate system, the complement 
\begin{equation}
\psi _{i_{1}\mathfrak{\ldots }i_{s}}^{j_{1}\mathfrak{\ldots }%
j_{s}}(s,\lambda _{\mathfrak{d}})=U_{[\mathfrak{b}_{1}}^{j_{1}}\ldots U_{%
\mathfrak{b}_{s}]}^{j_{s}}U_{i_{1}}^{\mathfrak{b}_{1}}\ldots U_{i_{s}}^{%
\mathfrak{b}_{s}}\prod\limits_{\mathfrak{c}\notin \left( \mathfrak{b}%
_{1},\ldots ,\mathfrak{b}_{s}\right) }\left( \lambda _{\mathfrak{d}}-\lambda
_{\mathfrak{c}}\right) \nonumber %
\end{equation}
with $\mathfrak{d}\in (\mathfrak{b}_{1},\ldots ,\mathfrak{b}_{s})$ 
has the meaning of the scaled oriented volume of a
parallelepiped stretched over $s$ linearly independent vectors 
$U_{\mathfrak{b}}^{j}$. 
The vectors $U_{\mathfrak{b}}^{j}$ are the eigenvectors, 
the explicit expressions for which we are searching.

The complement tensors $\psi_{i_{1} \ldots i_{r}}^{j_{1} \ldots j_{r}}(r, \lambda_{\mathfrak{d}})$ with $r < \mu_{\mathbf{H}}(\lambda_{\mathfrak{d}})=s$, however, 
vanish identically because the product sign in Eq.~(\ref{r ge 2}) contains $s-r \geq 1$ zero
multipliers $\lambda_{\mathfrak{d}}-\lambda_{\mathfrak{c}}$. 
In particular, 
\begin{align}
& s = 1 \Rightarrow \det \mathbf{C}(\lambda_{\mathfrak{d}}) = 0, \label{4.4.1} \\
& s = 2 \Rightarrow \det \mathbf{C}(\lambda_{\mathfrak{d}}) = \psi^j_i(1, \lambda_{\mathfrak{d}}) = 0, \label{4.4.2} \\
& s = 3 \Rightarrow \det \mathbf{C}(\lambda_{\mathfrak{d}}) = \psi^j_i(1, \lambda_{\mathfrak{d}}) = \psi^{j_1 j_2}_{i_1 i_2}(2,\lambda_{\mathfrak{d}}) = 0, \label{4.4.3}
\end{align}
and so on.

We have thus proven
\begin{theorem}
    Let $\mathbf{H}$ be an $n \times n$ Hermitian matrix, and let $\lambda_{\mathfrak{d}}$ be one of its eigenvalues of algebraic multiplicity $\mu_{\mathbf{H}}(\lambda_{\mathfrak{d}})=s > 1$. 
 In the diagonal coordinate system, the complement $\psi_{\mathfrak{b}_{1} \ldots \mathfrak{b}_{s}}^{\mathfrak{a}_{1} \ldots \mathfrak{a}_{s}}(s,\lambda _{\mathfrak{d}})$ of the minors of the characteristic matrix $\mathbf{C}(\lambda_{\mathfrak{d}})$ has non-vanishing components if and only if the index sets $(\mathfrak{a}_{1},\ldots, \mathfrak{a}_{s})$ and $(\mathfrak{b}_{1},\ldots, \mathfrak{b}_{s})$ coincide up to a permutation and $\mathfrak{d} \in (\mathfrak{b}_{1},\ldots, \mathfrak{b}_{s})$. Furthermore, 
     in all reference frames, the complements of order $r$ with $r<s$ vanish identically.
\end{theorem}



\subsection{Product of $A_{i}^{j}$ and $(\underset{s}{\det }\, \mathbf{A})_{i_{1}\mathfrak{\ldots }%
i_{r}}^{j_{1}\mathfrak{\ldots }j_{r}}$}

In search of an extension of the Eq. (\ref{basic eq adj}), we consider the product 
of the characteristic matrix and its reduced complement. The results of this subsection are 
valid for all matrices $\mathbf{A}$ and  given in general form.
\begin{lemma}
Contraction of the covariant index of an $n\times n$ matrix $\mathbf{A}$
with one of the contravariant indices of the reduced complement of order $r,s
$ of the minor determinants of order $n-s$ of matrix $\mathbf{A}$ can be
represented in two equivalent forms: through an alternating sum of minor
determinants of order $n-s+1$:
\begin{align}
A_{j_{r}}^{j}(\underset{s}{\det }\, \mathbf{A})_{i_{1}\mathfrak{\ldots }%
i_{r-1}i}^{j_{1}\mathfrak{\ldots }j_{r-1}j_{r}} &= 
\frac{1}{(n-s)!(n-s + 1)!}\delta_{i_{1} \ldots i_{r-1} i i_{s+1}\ldots i_{n}}^{j_{1}\ldots
j_{r-1}j_{s}j_{s+1}\ldots j_{n}} \nonumber \\
&~~~~~~~~~~~~~~~~~~~~~~~~~~~~~~\times (\det \mathbf{A})_{j_{s}j_{s+1}\ldots
j_{n}}^{j i_{s+1}\ldots i_{n}},  \label{great det}
\end{align}%
and an alternating sum of complements of order $s-1$ to
minor determinants of order $n-s+1$: 
\begin{align}
&A_{j_{r}}^{j}(\underset{s}{\det }\, \mathbf{A})_{i_{1}\mathfrak{\ldots }%
i_{r-1}i}^{j_{1}\mathfrak{\ldots }j_{r-1}j_{r}} = 
\frac{1}{(s-1)!(s - r)!}\delta_{i_{1}\ldots i_{r-1}j_{r+1}\ldots j_{s}i}^{k_{1}\ldots k_{s-1}j}
\nonumber \\
&~~~~~~~~~~~~~~~~~~~~~~~~~~~~~~~~~~~~~~~~~~~~~~~\times 
(\underset{s-1}{\det }\, \mathbf{A})_{k_{1}\ldots k_{s-1}}^{j_{1}\ldots
j_{r-1}j_{r+1}\ldots j_{s}}.  \label{great}
\end{align}
Equation (\ref{great}) is the generalization of Eq. (\ref{basic eq adj}).
\end{lemma}

\begin{proof}
The left side of Eq. (\ref{great det}) can be written in the form 
\begin{equation*}
\frac{1}{(n-s)!}\delta _{i_{1}\ldots i_{r-1}ii_{s+1}\ldots
i_{n}}^{j_{1}\ldots j_{r-1}j_{s}j_{s+1}\ldots
j_{n}}A_{j_{s}}^{j}A_{j_{s+1}}^{i_{s+1}}\ldots A_{j_{n}}^{i_{n}}.
\end{equation*}%
The product of the matrices $\mathbf{A}$ is antisymmetric in the covariant
and contravariant indices, and it can be replaced with the antisymmetric form 
$A_{[j_{s}}^{j}A_{j_{s+1}}^{i_{s+1}}$ $\ldots A_{j_{n}]}^{i_{n}}/(n-s+1)!$ to
get Eq. (\ref{great det}). 

To make antisymmetry of the product 
in the left side of Eq. (\ref{great}) explicit, we employ the generalized
Kronecker symbol: 
\begin{equation*}
\frac{1}{(n-s)!}\delta _{i_{1}\ldots i_{r-1}ii_{s+1}\ldots
i_{n}}^{j_{1}\ldots j_{r-1}j_{s}j_{s+1}\ldots j_{n}}\frac{1}{(n-s+1)!}\delta
_{k_{s}k_{s+1}\ldots k_{n}}^{j i_{s+1}\ldots
i_{n}}A_{j_{s}}^{k_{s}}A_{j_{s+1}}^{k_{s+1}}\ldots A_{j_{n}}^{k_{n}}.
\end{equation*}%
The combinatorial factor $1/(n-s+1)!$ accounts for the number of
permutations of the indices $k_{s}$, $k_{s+1},$ \ldots , $k_{n}$ involved in
the summation (cf. Eq.~(\ref{APP 0})). 
The first generalized Kronecker symbol of order $n - s + r$ and the second generalized Kronecker 
symbol of order $n-s+1$ can be substituted by products of two Levi-Civita symbols 
according to Eq~(\ref{APP 1}) and written as the generalized Kronecker symbols of order $n$:
\begin{align*}
& \frac{1}{(n-s)!(s - r)!}\delta _{i_{1}\ldots i_{r-1}l_{r+1}\ldots
l_{s}ii_{s+1}\ldots i_{n}}^{j_{1}\ldots j_{r-1}l_{r+1}\ldots
l_{s}j_{s}j_{s+1}\ldots j_{n}} \\
& 
~~~~~~~~~~~\times~
\frac{1}{(n-s+1)!(s-1)!} \delta
_{k_{1}\ldots k_{s-1}k_{s}k_{s+1}\ldots k_{n}}^{k_{1}\ldots
k_{s-1}ji_{s+1}\ldots i_{n}} A_{j_{s}}^{k_{s}}A_{j_{s+1}}^{k_{s+1}}\ldots A_{j_{n}}^{k_{n}}.
\end{align*}%
The combinatorial factors $1/(s-r)!$ and $1/(s-1)!$ account for the
number of permutations in the index sets $(l_{r+1},\ldots ,l_{s})$ and $(k_{1},\ldots ,k_{s-1})$, respectively. The exchange of
the Levi-Civita symbols with covariant indices between the first and second
generalized Kronecker symbols allows to write, for the left side of Eq.~(\ref%
{great}), 
\begin{align*}
& \frac{1}{(n-s)!(s - r)!}\delta _{i_{1}\ldots i_{r-1}l_{r+1}\ldots l_{s}i
i_{s+1}\ldots i_{n}}^{k_{1}\ldots k_{s-1}ji_{s+1}\ldots i_{n}}\\
& 
~~~~~~~~~~~\times~
\frac{1}{(n-s+1)!}\frac{1}{(s-1)!}\delta _{k_{1}\ldots k_{s-1}k_{s}k_{s+1}\ldots
k_{n}}^{j_{1}\ldots j_{r-1}l_{r+1}\ldots l_{s}j_{s}j_{s+1}\ldots j_{n}} 
~A_{j_{s}}^{k_{s}}A_{j_{s+1}}^{k_{s+1}}\ldots A_{j_{n}}^{k_{n}}.
\end{align*}%
The summation over the indices $i_{s+1},$ \ldots , $i_{n}$ generates an
order-$s$ generalized Kronecker symbol and removes the combinatorial factor $%
1/(n-s)!$. The summation over the indices $j_{s}$, $j_{s+1}$, \ldots , $j_{n}
$ and $k_{s}$, $k_{s+1}$, \ldots , $k_{n}$ results with the help of
Eqs.~(\ref{complement def}) and (\ref{define wf adj tensor}) in the expression
entering the right side of Eq.~(\ref{great}).
\end{proof}

Lemma 1 has consequences:

\begin{corollary}
The multiplication of complement $(\underset{s}\det \mathbf{A}) _{i_{1}
\ldots i_{s}}^{k_{1} \ldots k_{s}}$ 
by the product $A_{k_{1}}^{j_{1}}\ldots A_{k_{s}}^{j_{s}}
$ and the subsequent contraction of the indices $k_{1},$ $\ldots$, $k_{s}$
make the expression antisymmetric in the
indices. The product $A_{k_{1}}^{j_{1}}\ldots A_{k_{s}}^{j_{s}}$ turns to an order-$s$
minor determinant. The result of the multiplication can be
written as follows
\begin{equation}  \label{corollary 1}
A_{k_{1}}^{j_{1}}\ldots A_{k_{s}}^{j_{s}} (\underset{s}\det \mathbf{A})
_{i_{1} \ldots i_{s}}^{k_{1} \ldots k_{s}} =\delta_{i_{1}\ldots
i_{s}}^{j_{1}\ldots j_{s}} \det \mathbf{A}.
\end{equation}
This equation is the generalization of Eqs. (\ref{basic eq adj}) 
and (\ref{great}).
\end{corollary}

\begin{proof}
Let us start from the equation 
\begin{equation*}
A_{k_{1}}^{j_{1}}\ldots A_{k_{s}}^{j_{s}}(\underset{s}\det \mathbf{A})
_{i_{1} \ldots i_{s}}^{k_{1} \ldots k_{s}} =\frac {1}{(n-s)!}%
\delta_{i_{1}\ldots i_{s}i_{s+1}\ldots i_{n}}^{k_{1}\ldots
k_{s}k_{s+1}\ldots k_{n}}A_{k_{1}}^{j_{1}}\ldots
A_{k_{s}}^{j_{s}}A_{k_{s+1}}^{i_{s+1}}\ldots A_{k_{n}}^{i_{n}}.
\end{equation*}

The left side equals 
\begin{equation*}
(\det \mathbf{A})^{j_{1}\ldots j_{s}}_{k_{1}\ldots k_{s}}  (\underset{s}\det 
\mathbf{A}) _{i_{1} \ldots i_{s}}^{k_{1} \ldots k_{s}}, 
\end{equation*}
where the sum runs over the unordered index sets $(k_{1}, \ldots, k_{s})$.

Equation (\ref{det}) allows to write the right side of Eq. (\ref{corollary 1}) in the form 
\begin{equation*}
\frac{1}{(n-s)!}\delta_{i_{1}\ldots i_{s}i_{s+1}\ldots i_{n}}^{k_{1}\ldots
k_{s}k_{s+1}\ldots k_{n}}\delta_{k_{1}\ldots k_{s}k_{s+1}\ldots
k_{n}}^{j_{1}\ldots j_{s}i_{s+1}\ldots i_{n}}\frac{1}{n!}\det\mathbf{A}. 
\end{equation*}
The sum over the indices $k_{1},$ $\ldots,$ $k_{n}$ and $i_{s+1},$ $\ldots,$ 
$i_{n}$ gives Eq.~(\ref{corollary 1}).
\end{proof}

\begin{remark}
The matrix $\mathbf{A}$ degeneracy condition was not employed to prove
Theorem 2 or Corollary 1. 
By multiplying both parts of Eq.~(\ref{corollary 1}) by $(\mathbf{A}%
^{-1})^{l_{1}}_{j_{1}} \ldots (\mathbf{A}^{-1})^{l_{s}}_{j_{s}}$ and summing
by repeating indices, one finds (see also \cite{Prasolov:1994})
\begin{equation}
(\underset{s}\det \, \mathbf{A})_{i_{1}\ldots i_{s}}^{j_{1}\ldots j_{s}}=(\det
\mathbf{A}^{-1})_{i_{1}\ldots i_{s}}^{j_{1}\ldots j_{s}}\det \mathbf{A}.
\nonumber 
\end{equation}%
\end{remark}

\begin{remark}
For $r=s$, Eq. (\ref{great}) simplifies to
\begin{equation}
A_{j_{s}}^{j}(\underset{s}{\det }\, \mathbf{A})_{i_{1}\mathfrak{\ldots }%
i_{s-1}i}^{j_{1}\mathfrak{\ldots }j_{s-1}j_{s}}=\frac{1}{(s-1)!}\delta
_{i_{1}\ldots i_{s-1}i}^{k_{1}\ldots k_{s-1}j}(\underset{s-1}{\det } \,\mathbf{A%
})_{k_{1}\ldots k_{s-1}}^{j_{1}\ldots j_{s-1}}.  \nonumber 
\end{equation}%
\end{remark}

\begin{remark}
For $\mathbf{A}=\mathbf{C}(\lambda )$, $\mathrm{rank}(\mathbf{C}(\lambda )) = n - s$ and 
$\mu _{\mathbf{H}}(\lambda )=s \geq r$, 
\begin{equation}
A_{j_{r}}^{j}(\underset{s}{\det } \, \mathbf{A})_{i_{1}\mathfrak{\ldots }%
i_{r-1}i}^{j_{1}\mathfrak{\ldots }j_{r-1}j_{r}}=0. \label{CY = 0}
\end{equation}
Equation (\ref{great det}) yields this equality because the minor
determinants of order $n-s+1$ vanish. 
The right side of Eq. (\ref{great}) from the other hand is zero because, as previously shown, 
the complements of the minor determinants of order $n-s+1$ vanish.
\end{remark}


\subsection{Basic properties of $\psi_{i}^{j} (s,\lambda)$}

The reduced complement tensor contains an excessive number of components. 
The minimum version of $r = 1$ already has all of the information about the eigenvectors.
Using Eqs. (\ref{define wf adj tensor}) and (\ref{CY = 0}) for $\mu_{\mathbf{H}}(\lambda)=s > 1$, we obtain 
\begin{equation}
C_{k}^{j}(\lambda)\psi_{i}^{k}(s,\lambda)=0. \label{4.8}
\end{equation}
Equation (\ref{4.8}) is a fairly obvious equivalent of Eq.~(\ref{basic eq adj}) 
for characteristic matrices, 
from which eigenvectors for non-degenerate systems are derived in Sect. 3.

Contraction of $s - 1$ indices in Eq. (\ref{r ge 2}) yields the reduced order-one complement
\begin{equation*}
\psi_{\mathfrak{b}}^{\mathfrak{a}}(s,\lambda_{\mathfrak{d}})=(s-1)!\delta _{%
\mathfrak{b}}^{\mathfrak{a}}\prod\limits_{\mathfrak{c}\notin\left( \mathfrak{%
b}_{1},\ldots,\mathfrak{b}_{s}\right) }\left( \lambda_{\mathfrak{d}} -\lambda_{\mathfrak{c}%
}\right). 
\end{equation*}
The non-vanishing components are those with $\mathfrak{a} = \mathfrak{b} = \mathfrak{d} \in\left( \mathfrak{b}_{1},\ldots,\mathfrak{b}_{s}\right) 
$. There are $s$ such components in total.

Hermitian product that corresponds
to two distinct sets of eigenvalues $\lambda_{\mathfrak{b}}$
with $\mathfrak{b}\in(\mathfrak{b}_{1},\ldots,\mathfrak{b}_{s})$ and $%
\lambda _{\mathfrak{d}}$ with $\mathfrak{d}\in(\mathfrak{d}_{1},\ldots,%
\mathfrak{d}_{p})$ of algebraic multiplicities $\mu_{\mathbf{A}}(\lambda_{%
\mathfrak{b}})=s$ and $\mu_{\mathbf{A}}(\lambda_{\mathfrak{d}})=p$,
respectively, is equal to%
\begin{align*}
\sum_{\mathfrak{a}}&
\psi_{\mathfrak{g}}^{\mathfrak{a\ast}}(s,\lambda_{\mathfrak{b}})
\psi_{\mathfrak{h}}^{\mathfrak{a}}(p,\lambda_{\mathfrak{d}}) \\
&= (s-1)! (p-1)! \delta_{\mathfrak{h}}^{\mathfrak{g}}
\prod \limits_{\mathfrak{c}\notin\left( \mathfrak{b}_{1},\ldots,\mathfrak{b}_{s}\right) }
\left( \lambda_{\mathfrak{b}}-\lambda _{\mathfrak{c}}\right)
\prod \limits_{\mathfrak{e}\notin\left( \mathfrak{d}_{1},\ldots,\mathfrak{d}_{s}\right) }
\left( \lambda_{\mathfrak{d}}-\lambda_{\mathfrak{e}}\right) . 
\end{align*}
The index sets $(\mathfrak{b}_{1},\ldots,\mathfrak{b}_{s})$ and $(%
\mathfrak{d}_{1},\ldots,\mathfrak{d}_{p})$ do not have elements in common, so that
$\delta_{\mathfrak{h}}^{\mathfrak{g}} = 0$, which
provides the orthogonality. The contravariant components of the reduced 
complements define contravariant components of the eigenvectors, while the
lower lables enumerate these eigenvectors. 
For identical sets, $(\mathfrak{b}_{1},\ldots,\mathfrak{b}%
_{s})$ and $(\mathfrak{d}_{1},\ldots,\mathfrak{d}_{p})$, the orthogonality of eigenvectors is
provided by the Kronecker symbol as well.

In the initial coordinate system, 
\begin{align}
\psi_{i}^{j}(s,\lambda_{\mathfrak{b}}) & =
\sum_{\mathfrak{a},\mathfrak{g}\in\left( \mathfrak{b}_{1},\ldots,\mathfrak{b}_{s}\right) }
U_{\mathfrak{a}}^{j} \psi_{\mathfrak{g}}^{\mathfrak{a}}(s,\lambda_{\mathfrak{b}})
U_{i}^{\mathfrak{g}}  \notag \\
& =(s-1)!\sum_{\mathfrak{a}\in\left( \mathfrak{b}_{1},\ldots,\mathfrak{b}%
_{s}\right) }U_{\mathfrak{a}}^{j} U_{i}^{%
\mathfrak{a}}\prod\limits_{\mathfrak{c}\notin\left( \mathfrak{b}_{1},\ldots,%
\mathfrak{b}_{s}\right) }\left( \lambda_{\mathfrak{b}}-\lambda_{\mathfrak{c}%
}\right) .  \label{psi 3}
\end{align}
Suppose two different
sets $(\mathfrak{b}_{1},\ldots,\mathfrak{b}_{s})$ and $(\mathfrak{%
d}_{1},\ldots,\mathfrak{d}_{p})$ correspond to $s$- and $p$-fold
degenerate eigenvalues. The Hermitian product can be found to be
\begin{align}
\sum_{j}&\psi_{i}^{j\ast}(s,\lambda_{\mathfrak{b}})
\psi_{k}^{j}(p,\lambda_{\mathfrak{d}})  \nonumber \\ 
& =
\sum_{\mathfrak{a}\in\left( \mathfrak{b}_{1},\ldots,\mathfrak{b}_{s}\right)}
\sum_{\mathfrak{g}\in\left( \mathfrak{d}_{1},\ldots,\mathfrak{d}_{p}\right)}
\delta_{\mathfrak{g}}^{\mathfrak{a}}
U_{\mathfrak{a}}^{i} U_{k}^{\mathfrak{g}}
\prod\limits_{\mathfrak{c}\notin\left( \mathfrak{b}_{1},\ldots,%
\mathfrak{b}_{s}\right) }\left( \lambda_{\mathfrak{b}}-\lambda_{\mathfrak{c}%
}\right) \prod\limits_{\mathfrak{e}\notin\left( \mathfrak{d}_{1},\ldots,%
\mathfrak{d}_{p}\right) }\left( \lambda _{\mathfrak{d}}-\lambda_{\mathfrak{e}%
}\right) .  \notag
\end{align}
The reduced complements, considered as the eigenvectors with
regard to their contravariant components, appear to be orthogonal because%
\emph{\ }$\mathfrak{a}\neq\mathfrak{g}$\emph{\ }by the assertion\emph{\ }$(%
\mathfrak{b}_{1},\ldots,\mathfrak{b}_{s})\cap(\mathfrak{d}_{1},\ldots,%
\mathfrak{d}_{p})=\varnothing$.

The eigenvectors representing a single set of eigenvalues with $(\mathfrak{b}%
_{1},\ldots,\mathfrak{b}_{s})=(\mathfrak{d}_{1},\ldots,\mathfrak{d}_{p})$
are not orthogonal. Equation (\ref{psi 3}) shows that $n$ eigenvectors
numbered by the covariant indices form a superposition
of $s$ eigenvectors proportional to the matrix elements $U_{\mathfrak{a}}^{j}$. 

We have thus proven
\begin{theorem}
Let $\mathbf{H}$ be an $n \times n$ Hermitian matrix and let $\lambda_{\mathfrak{d}}$ be one of its eigenvalues of algebraic multiplicity $\mu_{\mathbf{H}}(\lambda_{\mathfrak{d}}) = s > 1$, related  to 
an index set $(\mathfrak{b}_1,\ldots,\mathfrak{b}_s) \in \Sigma_{n}$. 
The reduced complement $\psi_{\mathfrak{b}}^{\mathfrak{a}}(s,\lambda_{\mathfrak{d}})$ of the minor determinants of the characteristic matrix 
$\mathbf{C}^{\prime}(\lambda_{\mathfrak{d}})$
has non-vanishing components if and only if 
$\mathfrak{a} = \mathfrak{b} \in (\mathfrak{b}_1,\ldots,\mathfrak{b}_s)$ and $\mathfrak{d} \in (\mathfrak{b}_1,\ldots,\mathfrak{b}_s)$. 
The contravariant components of the reduced complement are the contravariant components of the non-vanishing eigenvectors of matrix $\mathbf{H}^{\prime}$, enumerated by the index $\mathfrak{b}$. The initial coordinate system associates the reduced complement $\psi_{i}^{j}(s,\lambda _{\mathfrak{d}})$ 
with the $n$ vectors labeled by the index $i \in \Sigma_{n}$. Only $s$ of them are linearly independent. These linearly independent vectors are the eigenvectors of matrix $\mathbf{H}$. The eigenvectors for the entire set of eigenvalues are orthogonal and can be normalized to unity.
\end{theorem}

The reduced complements produce normalized eigenvectors that form columns of 
a unitary matrix $\mathbf{V}$, which differs from the transformation matrix $\mathbf{U}$ by a 
multiplication on the diagonal matrix of phase factors: 
$\mathbf{V} = \mathbf{U} \Phi$, where $\Phi = \mathrm{diag} (e^{i\phi_1},\ldots,e^{i\phi_n})$. 
This observation can be used to construct a neutrino mixing matrix based on the neutrino mass matrix.
An alternate approach to the same problem can be found in \cite{Krivoruchenko:2024}. 

\subsection{Trace identity for $\psi_{i}^{j} (s,\lambda)$}

The operation of determining the trace of a matrix product is explicitly covariant with respect to the 
coordinate system transformations. Minor determinants are also covariant objects, but unlike trace o
operations, their explicit covariance gets lost during component calculations, e.g., by the Gauss method. 
By employing analytical methods, the trace procedure offers particular advantages. Furthermore, matrices 
with multicomponent indexes may be manipulated more easily. For instance, consider the Green function of 
quarks, which has both a bispinor and a color index. The summation of each of the indices is carried out
independently, allowing for advancements in the structural analysis of the equations under consideration.


An explicit expression for the reduced order-one complements in terms of the powers and the trace of powers of matrix $\mathbf{A}$ is provided by

\begin{theorem}
The reduced order-one complement of the minor determinants of order $n-s$ of an $n \times n$ matrix $\mathbf{A}$ 
is given, in terms of the powers and the trace of powers of the matrix $\mathbf{A}$, by equation
\begin{equation}
(\underset{s}\det \mathbf{A})^{j}_{i}
= \sum_{r=0}^{n-s} ( -1) ^{n-s} \left( \mathbf{A}^{r}\right)_{i}^{j}
\sum_{k_{1},\ldots,k_{n-s-r}} 
\prod\limits_{l=1}^{n-s-r}\frac{\left( -1\right) ^{k_{l}}}
{k_{l}! l^{k_{l}}}\mathrm{Tr}[\mathbf{A}^{l}]^{k_{l}},  \label{main eq}
\end{equation}
where $k_{l}$ are non-negative integers satisfying the constraint%
\begin{equation}
r+\sum_{l=1}^{n-s-r}lk_{l}=n-s.  \label{main eq constraint}
\end{equation}
\end{theorem}

\begin{proof}
We write Eq. (\ref{basic eq adj}) in the form
\begin{equation}
\det \mathbf{A} \left( \mathbf{A}^{-1}\right) _{i}^{j}=%
\frac{1}{(n-1)!}\delta
_{ii_{2}\ldots i_{n}}^{jj_{2}\ldots j_{n}} A_{j_{2}}^{i_{2}}\ldots A_{j_{n}}^{i_{n}}  \label{lhside}
\end{equation}
and replace $\mathbf{A}\rightarrow\mathbf{D}=\mathbf{I}+z\mathbf{A}$, where $%
\mathbf{I}$ is the unit matrix, and $z$ is a complex number. The right
side of the equation, being integrated over $z$ counterclockwise around $z=0,$ gives
\begin{align*}
&\frac{1}{2\pi i}\oint \frac{dz}{z^{n-s+1}}\frac{1}{(n-1)!}\delta
_{ii_{2}\ldots i_{n}}^{jj_{2}\ldots j_{n}}D_{j_{2}}^{i_{2}}\ldots
D_{j_{n}}^{i_{n}} \\
&~~~~~~=\frac{1}{2\pi i}\oint \frac{dz}{z^{n-s+1}}\frac{1}{(n-1)!}\delta
_{ii_{2}\ldots i_{n}}^{jj_{2}\ldots j_{n}} \\
&~~~~~~\times \sum_{r=0}^{n-1}z^{n-1-r}\frac{%
(n-1)!}{r!(n-1-r)!} \delta _{j_{2}}^{i_{2}}\ldots \delta
_{j_{r+1}}^{i_{r+1}}A_{j_{r+2}}^{i_{r+2}}\ldots A_{j_{n}}^{i_{n}}.
\end{align*}%
By performing summation over the repeated indices and performing the integration over $z$, we obtain
\begin{align*}
&\frac{1}{2\pi i}\oint \frac{dz}{z^{n-s+1}}\sum_{r=0}^{n-1}z^{n-1-r}\frac{1%
}{(n-1-r)!}\delta _{ii_{r+2}\ldots i_{n}}^{jj_{r+2}\ldots
j_{n}}A_{j_{r+2}}^{i_{r+2}}\ldots A_{j_{n}}^{i_{n}} \\
&~~~~~~=\frac{1}{2\pi i}\oint \frac{dz}{z^{n-s+1}}\sum_{r=0}^{n-1}z^{n-1-r}(%
\underset{r+1}{\det }\mathbf{A)}_{i}^{j}=(\underset{s}{\det }\mathbf{A)}%
_{i}^{j}.
\end{align*}

The left side of Eq. (\ref{lhside}) can be transformed as
follows:%
\begin{equation}
(\det \mathbf{D} ) \mathbf{D}^{-1}=\exp\left( \mathrm{Tr}%
\ln\left( \mathbf{I}+z\mathbf{A}\right) \right) \left( \mathbf{I}+z\mathbf{A}%
\right) ^{-1}.  \nonumber 
\end{equation}
The power series expansion around $z=0$ gives 
\begin{eqnarray}
\exp\left( \mathrm{Tr}\ln\left( \mathbf{I}+z\mathbf{A}\right) \right) 
&=& \prod \limits_{l=1}^{\infty}\sum_{k_{1},\ldots,k_{\infty}}\frac{%
(-1)^{(l+1)k_{l}}}{l^{k_{l}}k_{l}!}z^{lk_{l}}\mathrm{Tr[}\mathbf{A}^{l}%
\mathrm{]}^{k_{l}},  \label{expo1} \\
\left( \mathbf{I}+z\mathbf{A}\right) ^{-1} 
&=& \sum_{r=0}^{\infty}(-1)^{r}z^{r}\mathbf{A}^{r}  \label{expo2}
\end{eqnarray}
where $0\leq k_{l}<+\infty$. Substituting these expressions for the contour
integral in Eq. (\ref{lhside}), we reproduce the representation (\ref{main
eq}) with the parameters satisfying the constraint (\ref{main eq constraint}%
). These constraints imply, in particular, $r\leq n-s$, $l\leq n-s-r$, and $%
k_{l}=0$ for $l>n-s-r$. The summation runs over all non-negative solutions $%
k_{1}$, \ldots, $k_{n-s-r}$ of Eq. (\ref{main eq constraint}).
\end{proof} 

    The expansion 
\begin{equation}
    \exp \left( \sum_{l=1}^{\infty}x_{l}\frac{z^l}{l!} \right ) = \sum_{n=0}^{\infty}\frac{z^n}{n!}
    B_{n}(x_{1},\ldots,x_{n}) \nonumber %
\end{equation}
determines the complete exponential Bell polynomials $B_{n}(x_{1},\ldots,x_{n})$ \cite{Bell:1934}. 
By comparing the decomposition coefficients at the same powers of $z$,
one gets 
\begin{equation}
    \frac{1}{n!}B_{n}(x_{1},\ldots,x_{n}) = \prod_{l=1}^{n}\left( \sum_{k_{l} = 0}^{n} \frac{x_{l}^{k_{l}}}{k_{l}!l!^{k_{l}}}\right) \nonumber %
\end{equation}
under the condition of
\begin{equation}
\sum_{l=1}^{n}lk_{l} = n.   \nonumber 
\end{equation}

\begin{corollary}
Using the above formulas, Eq. (\ref{main eq}) can be written in the form
\begin{equation}
(\underset{s}\det\,  \mathbf{A})^{j}_{i} = 
\sum_{r=0}^{n-s} \frac{(- 1) ^{n-s}}{(n - s - r)!} \left( \mathbf{A}^{r}\right)_{i}^{j}  B_{n - s - r}(x_{1},\ldots,x_{n - s - r}), 
\label{main eq plus}
\end{equation}
where $x_{l} = - (l-1)!\mathrm{Tr}[\mathbf{A}^{l}]$.

\end{corollary}

\begin{example}
Particular cases of Eqs. (\ref{main eq}) and (\ref{main eq plus})
\end{example}

Using the Maple symbolic computation
software package \cite{Maple} we have verified Eqs. (\ref{main eq}) and (\ref{main eq plus}) for $n = 3$ with $s=1,2$ and $n = 4$ with $s=1,2,3$. 

The right side of Eq. (\ref{main eq}) for $s = 1$ reproduces the adjugate
matrix in terms of the trace expansion \cite{Kondratyuk:1992,Krivoruchenko:2016}. 
Moreover, the expansion depends only on the difference $n - s$, so 
the adjugate matrix converts to the expressions (\ref{main eq}) and (\ref{main eq plus}) after replacing $n - 1 \to n - s$.
For simple cases, we obtain
\begin{center}
$%
\begin{array}{cl}
\hline
\hline
n-s & ~~~~~~~~~~~~~~~~~~~~~~~~(\underset{s}\det \mathbf{A})^{j}_{i} \\ 
\hline
1 & \mathrm{Tr}[\mathbf{A}]-\mathbf{A} \\ 
2 & \frac{1}{2}\left( \mathrm{Tr}[\mathbf{A}]^{2}-\mathrm{Tr}[\mathbf{A}%
^{2}]\right) -\mathbf{A}\mathrm{Tr}[\mathbf{A}]+\mathbf{A}^{2} \\ 
3 & \frac{1}{6}\left( \mathrm{Tr}[\mathbf{A}]^{3}-3\mathrm{Tr}[\mathbf{A}]%
\mathrm{Tr}[\mathbf{A}^{2}]+2\mathrm{Tr}[\mathbf{A}^{3}]\right) \\
& ~~~~~~~~-\frac{1}{2}%
\mathbf{A}\left( \mathrm{Tr}[\mathbf{A}]^{2}-\mathrm{Tr}[\mathbf{A}%
^{2}]\right) +\mathbf{A}^{2}\mathrm{Tr}[\mathbf{A}]-\mathbf{A}^{3} \\
\hline
\hline
\end{array}
$
\end{center}

\begin{example}
Dirac equation
\end{example}

It is informative to demonstrate this formalism with the free Dirac equation, which has a twofold 
degeneration of energy levels corresponding to distinct spin states of an electron \cite{Bjorken:1964}. The 
$\gamma$-matrices have a size of $n=4$, and the algebraic multiplicity of levels equals $s=2$. 
In the standard representation, the Hamiltonian is represented as 
$\mathbf{H} = \mbox{\boldmath$\alpha$} \mathbf{p} + \beta m$, 
with $\mbox{\boldmath$\alpha$} = \gamma_{0}\mbox{\boldmath$\gamma$}$ and $\beta=\gamma_{0}$. 
The eigenvalues are 
$\lambda_{1} =\lambda_{2}=-\lambda_{3}=-\lambda_{4} = E_{\mathbf{p}} \equiv \sqrt{\mathbf{p}^{2}+m^{2}}$. 
The characteristic matrix equals $\mathbf{C}(\lambda) = \lambda - \mathbf{H}$, 
$\mathrm{rank}(\mathbf{C}(\lambda_{\mathfrak{a}})) = 2$ for $\mathfrak{a} = 1,\ldots,4$. 
We need the following values: 
\begin{align*}
\mathbf{C}(\lambda _{\mathfrak{a}})& =\gamma _{0}\left( \hat{p}-m\right) , \\
\mathbf{C}^{l}(\lambda _{\mathfrak{a}})& =2^{l-1}\lambda _{\mathfrak{a}%
}^{l-1}\mathbf{C}(\lambda _{\mathfrak{a}}), \\
\mathrm{Tr}[\mathbf{C}^{l}(\lambda _{\mathfrak{a}})]& =2^{l+1}\lambda _{%
\mathfrak{a}}^{l},
\end{align*}%
where $\hat{p}=p_{\mu }\gamma ^{\mu }$ with $p_{0}=\lambda _{\mathfrak{a}}$
and $l=1,2,\ldots $. Equations (\ref{main eq}) and (\ref{main eq plus}) yield

\begin{equation*}
\psi _{i}^{j}(2,\lambda _{\mathfrak{a}})=2\lambda _{\mathfrak{a}}\left(
\left( \hat{p}+m\right) \gamma _{0}\right) _{i}^{j}.
\end{equation*}%
The fundamental requirement (\ref{4.8}) 
is successfully fulfilled.

To validate Eqs. (\ref{4.4.2}), we find additionally
\begin{eqnarray*}
\det \mathbf{C}(\lambda ) &=&(\lambda ^{2}-E_{\mathbf{p}}^{2})^{2}, \\
\mathbf{C}^{-1}(\lambda ) &=&\frac{\hat{p}+m}{\lambda ^{2}-E_{\mathbf{p}}^{2}%
}\gamma _{0},
\end{eqnarray*}%
where $p_{0}=\lambda $. It is seen that $\det \mathbf{C}(\lambda _{\mathfrak{a}})=0$ and 
\begin{equation*}
\psi _{i}^{j}(1,\lambda _{\mathfrak{a}})=(\mathrm{adj}~\mathbf{C}(\lambda _{%
\mathfrak{a}}))_{i}^{j}=\lim_{\lambda \rightarrow \lambda _{\mathfrak{a}}}(%
\mathbf{C}^{-1}(\lambda )\det \mathbf{C}(\lambda ))_{i}^{j}=0.
\end{equation*}

Equations (\ref{4.4.2}) can also be verified with trace identities for the determinant and adjugate matrix.

The subscript $i$ counts the eigenvectors. This index can be
contracted with any matrix. In this situation, one may cancel the
multiplication by $\gamma _{0}$ to get the conventional projection operator
for states with energy $p_{0}$ and momentum $\mathbf{p}$: 
\begin{equation*}
\Pi_{\mathbf{p}}=\frac{\hat{p}+m}{2m}.
\end{equation*}%
The projection operator can act on any bispinors to build the eigenvectors. For instance, bispinors of
electrons, polarized along the $z$-axis in the rest frame and boosted in the
direction of the momentum $\mathbf{p}$, have the standard form: 
$u(\mathbf{p},s) = N_{\mathbf{p}}\Pi_{\mathbf{p}} u(\mathbf{p=0},e_z)$, where 
$u^{\mathrm{T}}(\mathbf{p=0}, e_z) $ $ =(1,0,0,0)$ for $\mathfrak{a}=1$ and 
$u^{\mathrm{T}}(\mathbf{p=0}, $ $ -e_z) $ $ =(0,1,0,0)$ for $\mathfrak{a}=2$. 
Here, $e_{z} = (0,0,0,1)$ is the polarization four-vector in the rest frame 
and $s_{\mu}$ is the polarization four-vector in the boosted frame. 
For the negative energy solutions, bispinors take the form 
$v(\mathbf{p},s) $ $  = N_{\mathbf{p}}\Pi_{\mathbf{p}}v(\mathbf{p=0},e_z)$, where 
$v^{\mathrm{T}}(\mathbf{p=0},e_z) $ $ =(0,0,0,1)$ for $\mathfrak{a}=3$ and 
$v^{\mathrm{T}}(\mathbf{p=0},-e_z) $ $ =(0,0,1,0)$ for $\mathfrak{a}=4$.
The normalization factors can be found from the equations $\bar{u}(p,s)u(p,s)= - \bar{v}(p,s)v(p,s) = 1$
to give $N_{\mathbf{p}}=\sqrt{\frac{m}{E_{\mathbf{p}}+m}}$.

Charge conjugation, which is an issue of physical interpretation, is normally
applied to the negative energy solutions ($\mathfrak{a}=3,4$). In
such cases, this requires substituting $s_{\mu} \rightarrow -s_{\mu}$ and $p_{\mu} \rightarrow -p_{\mu}$
in the above formulas to identify the eigenvectors $v(p,s)$ with the positively charged positrons.

The choosing of bispinors in the rest system only affects the phase of the eigenvectors built with the use of projection operators on states with certain four-momentum and spin polarization \cite{Krivoruchenko:1994}. Spin degeneracy opens the door to the possibility of mixing the eigenvectors by spatial rotations. This arbitrary element pervades all systems with 
degenerate energy spectra. 

\section{Conclusions}

The present investigation is based on Eq. (\ref{basic eq adj}), which represents an eigenvalue equation for $n \times n$ matrices with rank $n - 1$. The material of Sect. 4 
is devoted to generalizing Eq.~(\ref{basic eq adj}) to the case of matrices with higher degeneracy (Theorem 4). Such a demand occurs while working with the Dirac equation and in numerous other applications. We have limited ourselves to examining Hermitian matrices to simplify the whole procedure, as well as because Hermitian matrices are almost exclusively used in quantum physics: the Hamiltonian that determines the evolution of quantum systems must be Hermitian in order to preserve probability. 

The novel formalism maintains the structure and symmetries of the eigenvalue equations at intermediate stages of calculations, enabling in-depth analysis of solution properties and expanding the scope of analytical techniques.

\section*{Acknowledgements}
This work was supported by the Russian Science Foundation, Project No. 23-22-00307.

\section*{Appendix A. The generalized Kronecker symbol}
\renewcommand{\theequation}{A.\arabic{equation}}
\setcounter{equation}{0}

In this appendix, we list key properties of the generalized Kronecker
symbol, used to derive formulas in the main sections.

\begin{definition}
The generalized Kronecker
symbol of order $s$ in an $n$-dimensional Euclidean space is expressed in terms of the
antisymmetric Levi-Civita symbol of order $n$:
\begin{equation}
\delta_{i_{1}\ldots i_{s}}^{j_{1}\ldots j_{s}}=\frac{1}{(n-s)!}\epsilon
^{j_{1}\ldots j_{s}k_{s+1}\ldots k_{n}}
\epsilon_{i_{1}\ldots i_{s}k_{s+1}\ldots k_{n}}.  \label{APP 1}
\end{equation}
The indices $k_{s+1},$ \ldots, $k_{n}$, for which summation is carried out, can each be
assigned independently any one of the values $1, \ldots, n$. 
\end{definition}
Using $s=1$,
one gets to the standard definition of the ordinary Kronecker symbol $%
\delta_{i}^{j}$. The numerical values of the generalized Kronecker symbol
are independent of the specific coordinate system in use due to the invariance
of the Levi-Civita symbol.

The generalized Kronecker symbol can be used to define antisymmetrized tensors:
\begin{definition}
Let $T_{i_{1}\ldots i_{s}}$ be an order-$s$ tensor. Its antisymmetrized
partner is defined by
\begin{equation}
T_{[i_{1}\ldots i_{s}]}=\delta _{i_{1}\ldots i_{s}}^{j_{1}\ldots
j_{s}}T_{j_{1}\ldots j_{s}}.  \label{APP 0}
\end{equation}%
The summation is performed over the indices $j_{1}$, $\ldots$,
$j_{s}$ independently or, which is equivalent, over the ordered index sets $(j_{1},\ldots
,j_{s}) \subseteq \Sigma_{n}.$ 
\end{definition}
Antisymmetrizing a completely antisymmetric tensor of order $s$ results in its
multiplication by $s!$.

\begin{remark}
The order of the generalized Kronecker symbol is reduced by contracting its 
covariant and contravariant indices:
\begin{equation}
\delta_{i_{1}...i_{s}j_{s+1}\ldots j_{p}}^{j_{1}...j_{s}j_{s+1}\ldots j_{p}}=%
\frac{(n-s)!}{(n-p)!}\delta_{i_{1}...i_{s}}^{j_{1}...j_{s}}.  \label{APP 3}
\end{equation}
\end{remark}

The summation involves $(p-s)!\mathcal{C}_{p-s}^{n-s}=(n-s)!/(n - p)!$
unordered samples $j_{s+1}$, \ldots, $j_{p}$ selected from $n-s$ numbers of
the set $\Sigma_{n}\backslash(j_{1},\ldots,j_{s})$. $\mathcal{C}_{k}^{n}$ is the binomial coefficient. 
The signs of the terms
are unaffected by simultaneous permutations of the
contravariant and covariant indices, hence the unique factor $(n-s)!/(n - p)!$ appears.

Equation (\ref{APP 3}) also follows directly from
definition (\ref{APP 1}).

\begin{remark}
Contraction of $s$ covariant and $s$ contravariant indices in the product of two
generalized Kronecker symbols of order $s$ and order $p$ 
yields a generalized Kronecker symbol of order $p - s$:
\begin{equation}
\delta_{i_{1}\ldots i_{s}}^{j_{1}\ldots j_{s}}\delta_{j_{1}\ldots
j_{s}i_{s+1}\ldots i_{p}}^{i_{1}\ldots i_{s}j_{s+1}\ldots j_{p}}=s!\frac {%
(n-p+s)!}{(n-p)!}\delta_{i_{s+1}\ldots i_{p}}^{j_{s+1}\ldots j_{p}}.
\label{APP 4}
\end{equation}
\end{remark}

    Let us comment on this equation. The unordered index sets $(j_{1},\ldots, j_{s})$ and $(i_{1},\ldots,i_{s})$ 
are the same. This match is provided by the first generalized Kronecker symbol. 
The second generalized Kronecker symbol ensures that the unordered index sets $(i_{1},\ldots,i_{s}j_{s+1}, $ $\ldots, j_{p})$ 
and $(j_{1},\ldots, j_{s}i_{s+1}, \ldots, i_{p})$ coincide. Since the index sets in the first generalized Kronecker symbol are an integral part of the index sets
in the second, the unordered index sets 
$(j_{s+1}, \ldots, j_{p})$ 
and $(i_{s+1}, \ldots, i_{p})$ are also the same. This fact explains 
the presence of the generalized Kronecker symbol on the right side of the equation.
As noted earlier, independent summation over indices is equivalent to summation over their ordered sets. 
For fixed $(j_{s+1}, \ldots, j_{p})$, 
the ordered sets $(j_{1},\ldots,j_{s})$ can be formed in 
$s!\mathcal{C}_{n-p}^{n-p+s}=(n-p+s)!/(n-p)!$ various ways. 
Summation by $i_{s+1}$, $\ldots$, $i_{p}$ is restricted by the first generalized Kronecker symbol and gives for the right side of the equation an additional coefficient of $s!$.

\begin{remark}
Contraction of $s$ covariant indices of the generalized Kronecker symbol of order $s$ 
with $s$ contravariant indices of the generalized Kronecker symbol of order 
$p$ yields a generalized Kronecker symbol of order $p$:
\begin{equation}
\delta_{k_{1} \ldots k_{s}}^{j_{1}\ldots j_{s}}
\delta_{i_{1}\ldots i_{s}i_{s+1} \ldots i_{p}}^{k_{1} \ldots k_{s}j_{s+1} \ldots j_{p}}
=s!\delta_{i_{1}\ldots i_{s}i_{s+1}\ldots i_{p}}^{j_{1}\ldots j_{s}j_{s+1}\ldots
j_{p}}.  \nonumber 
\end{equation}
\end{remark}

The unordered sets of the covariant and contravariant indices of the first
and second generalized Kronecker symbols on the left side of the equation coincide. 
So the sets of the covariant and contravariant indices
on the right side coincide as well. This
observation allows to fix the tensor structure of the equation. To determine the
coefficient, it is sufficient to contract the indices
$j_{1}$, $\ldots$, $j_{s}$ and $i_{1}$, $\ldots$, $i_{s}$, respectively, and use Eq. (\ref{APP 4}).

\begin{theorem}\textbf{The generalized Cauchy-Binet formula}

Let $\mathbf{A}$ be an $p\times m$ matrix,
and $\mathbf{B}$ be an $m\times q$ matrix, and $p\leq m$, $q\leq m$, then%
\begin{equation} 
( \det  \mathbf{AB} )_{I}^{J}=\sum_{K\subseteq (1,\ldots ,m)}
( \det \mathbf{A} )_{K}^{J} (\det \mathbf{B} )_{I}^{K},
\label{APP6}
\end{equation}
where $J=(j_{1},\ldots ,j_{s})\subseteq (1,\ldots ,p)$, $I=(i_{1},\ldots
,i_{s})\subseteq (1,\ldots ,q)$ and the summation runs over an unordered index
sets $K=(k_{1},\ldots ,k_{s})\subseteq (1,\ldots ,m)$, which specifies columns
and rows of matrices $\mathbf{A}$ and $\mathbf{B}$, respectively.
\end{theorem}

\begin{proof}
The minor determinant of $\mathbf{AB}$ can be transformed as follows:%
\begin{align*}
( \det \mathbf{AB} )_{I}^{J}& =\frac{1}{s!}\delta _{l_{1}\ldots
l_{s}}^{j_{1}\ldots j_{r}}\delta _{i_{1}\ldots i_{s}}^{m_{1}\ldots
m_{r}}\left( \mathbf{AB}\right) _{m_{1}}^{l_{1}}\ldots \left( \mathbf{AB}\right)
_{m_{s}}^{l_{s}} \\
& =\frac{1}{(s!)^{3}}\delta _{l_{1}\ldots l_{s}}^{j_{1}\ldots j_{r}}\delta
_{i_{1}\ldots i_{s}}^{m_{1}\ldots m_{r}}A_{[k_{1}}^{l_{1}}\ldots
A_{k_{s}]}^{l_{s}}B_{[m_{1}}^{k_{1}}\ldots B_{m_{s}]}^{k_{s}} \\
& =\frac{1}{s!}A_{[k_{1}}^{j_{1}}\ldots
A_{k_{s}]}^{j_{s}}B_{[i_{1}}^{k_{1}}\ldots B_{i_{s}]}^{k_{s}}.
\end{align*}
Independent summation over the repeated indices yields $s!$ identical entries in each term of the associated unordered index set $(k_{1},\ldots,k_{s})$. The sum can be conducted over the unordered index sets $K=(k_{1},\ldots,k_{s})$ with a weight of $s!$. Applying Eq. (\ref{def 3}) yields Eq. (\ref{APP6}). The number of terms included in Eq.~(\ref{APP6}) equals $\mathcal{C}_{s}^{m}$. 
\end{proof}

\begin{corollary} \textbf{The Cauchy-Binet formula}
    
    For $p=q=s=n \leq m$ and $J = I = \Sigma_{n}$, one gets 
\begin{equation}
\det \mathbf{AB} = \sum_{K\subseteq (1,\ldots ,m)}
(\det \mathbf{A} )_{K}^{J} (\det \mathbf{B})_{I}^{K}.
\label{APP7}
\end{equation}
\end{corollary}

\section*{Appendix B. Diagonal sum of minor determinants}
\renewcommand{\theequation}{B.\arabic{equation}}
\setcounter{equation}{0}

We express the determinant of matrix $\mathbf{A}$ as 
\begin{equation}
\det \mathbf{A}  =\frac{1}{n!}\delta_{i_{1}\ldots
i_{n}}^{j_{1}\ldots j_{n}}A_{j_{1}}^{i_{1}}\ldots A_{j_{n}}^{i_{n}},
\label{lhside det}
\end{equation}
and then substitute $\mathbf{A}\rightarrow \mathbf{D}=\mathbf{I}+z\mathbf{A}$ as in Sect. 4.4.

The right side of the equation, being integrated over $z$ counterclockwise around $z=0$ and with a weight of $1/z^{n - s + 1}$, gives a diagonal sum (trace) of minor determinants of order $n - s$: 
\begin{align*}
\frac{1}{2\pi i}\oint \frac{dz}{z^{n-s+1}}\frac{1}{n!}\delta_{i_{1}\ldots
i_{n}}^{j_{1}\ldots j_{n}}D_{j_{1}}^{i_{1}}\ldots D_{j_{n}}^{i_{n}} & =
\frac{1}{(n - s)!}\delta _{i_{s+1}\ldots i_{n}}^{j_{s+1}\ldots
j_{n}}A_{j_{s+1}}^{i_{s+1}}\ldots A_{j_{n}}^{i_{n}} \\
& = \sum_{J} (\det \mathbf{A})_{J}^{J},
\end{align*}
where $J$ runs over the unordered index sets. 
The power series expansion of $\det (\mathbf{I}+z\mathbf{A}%
) =\exp\left( \mathrm{Tr}\ln\left( \mathbf{I}+z\mathbf{A}\right)
\right) $ around $z=0$ has the form of Eq. (\ref{expo1}). 

By substituting the left side of Eq. (\ref{lhside det}) in the form of a power series to the contour integral and performing the integration, we obtain
\begin{theorem}
The diagonal sum of minor determinants of order $n - s$ of an $n\times n$
matrix $\mathbf{A}$ is expressed in terms of traces of its powers: 
\begin{equation}
\frac{1}{(n-s)!}\delta_{i_{s+1}\ldots i_{n}}^{j_{s+1}\ldots
j_{n}}A_{j_{s+1}}^{i_{s+1}}\ldots
A_{j_{n}}^{i_{n}} =  \sum_{k_{1},\ldots,k_{n-s}}\prod%
\limits_{l=1}^{n-s}\frac{\left( -1\right) ^{k_{l}+1}}{l^{k_{l}}k_{l}!}%
\mathrm{Tr}[\mathbf{A}^{l}]^{k_{l}},  \label{main det eq}
\end{equation}
with parameters $k_{l}$ satisfying the constraint 
\begin{equation}
\sum_{l=1}^{n-s}lk_{l}=n-s.  \label{main det eq constraint}
\end{equation}
\end{theorem}
The summation over $l$ runs for $l\leq n-s$, so that $k_{l}=0$ for $l>n-s$.
The sum includes all non-negative solutions, $k_{1}$, \ldots, $k_{n-s}$
of Eq. (\ref{main det eq constraint}). 

\begin{corollary}
Equation (\ref{main det eq}) can be written in terms of the complete exponential Bell polynomial: 
\begin{equation} \label{main det eq plus}
\sum_{J} (\det \mathbf{A}) _{J}^{J}  =  \frac {(-1)^{n-s}}{(n-s)!}B_{n-s}(x_{1},\ldots,x_{n-s}), 
\end{equation}
where  $x_{l}=-(l-1)!\mathrm{Tr}[\mathbf{A}^{l}]$ and 
 $J$ runs over the unordered index sets $(j_{s+1},\ldots,j_{n}) \subseteq \Sigma_{n}$.
\end{corollary}

The determinant form of the complete exponential Bell polynomials enables presenting the determinant of a matrix $\mathbf{A}$ as the determinant of a matrix $\mathbf{B}$, whose elements are the traces of the powers of $\mathbf{A}$. This property is exploited 
to solve the Faddeev--Le Verrier recursion \cite{Faddeyev:1949,Frame:1949,LeVerrier:1840} via the determinant of $\mathbf{B}$ \cite{Reed:1978, Brown:1994}. The determinant form is also accessible for Eqs. (\ref{main eq}), (\ref{main eq plus}), (\ref{main det eq}), and (\ref{main det eq plus}).

\begin{example}
Particular cases of Eqs. (\ref{main det eq}) and (\ref{main det eq plus})
\end{example}

 Using the Maple symbolic computation
software package \cite{Maple} we have verified Eqs. (\ref{main det eq}) and (\ref{main det eq plus}) for $n = 3$ with $s=1,2$ and $n = 4$ with $s=1,2,3$. 

The right side of Eqs. (\ref{main det eq}) and (\ref{main det eq plus}) converts to a determinant after replacing $n - s \to n$. For other simple cases of $n - s = 1, \ldots, 4$, we
obtain:
\begin{center}
$%
\begin{array}{cl}
\hline
\hline
\vspace{0.5mm}
n-s & ~~~~~~~~~~~~~\sum_{J} (\det \mathbf{A}) _{J}^{J} \\ 
\hline
1 & ~~~~\mathrm{Tr}[\mathbf{A}]\\ 
2 & \frac{1}{2!} \left( \mathrm{Tr}[\mathbf{A}]^{2}-\mathrm{Tr}[\mathbf{A}^{2}] \right) \\ 
3 & \frac{1}{3!} \left( \mathrm{Tr}[\mathbf{A}]^{3}-3\mathrm{Tr}[%
\mathbf{A}]\mathrm{Tr}[\mathbf{A}^{2}]+2\mathrm{Tr}[\mathbf{A}^{3}]\right) \\ 
4 & \frac{1}{4!} \left( \mathrm{Tr}[\mathbf{A}]^{4}-6\mathrm{Tr}[%
\mathbf{A}^{2}]\mathrm{Tr}[\mathbf{A}]^{2}+3\mathrm{Tr}[\mathbf{A}^{2}]^{2} \right. \nonumber \\
&~~~~~~~~~~~~~\left. +~8 \mathrm{Tr}[\mathbf{A}^{3}]\mathrm{Tr}[\mathbf{A}] 
- 6\mathrm{Tr}[\mathbf{A}^{4}]\right)\\
\vdots  & ~~~~~~~~ \vdots \\
n & ~~~~~ \det \mathbf{A}\\
\hline
\hline
\end{array}
$
\end{center}

\end{document}